\documentclass[conference]{IEEEtran}
\IEEEoverridecommandlockouts
\usepackage[bookmarks, colorlinks=true, allcolors=blue]{hyperref}
\usepackage{amssymb}
\usepackage{amsmath}
\usepackage{amsthm}
\usepackage{bbm}
\usepackage[utf8]{inputenc}
\usepackage{dsfont}
\usepackage{mathtools}

\usepackage{centernot}
\def\openone{\leavevmode\hbox{\small1\kern-3.8pt\normalsize1}}

\def\mod{\,\, {\rm mod}\,\,}

\def\11{\mathbb{I}}

\newtheorem{definition}{Definition}[section]

\newtheorem{lemma}[definition]{Lemma}

\newtheorem{theorem}[definition]{Theorem}

\newcommand{\tr}{\mathop{\rm Tr}\nolimits}

\usepackage{pst-node}
\usepackage{tikz-cd}

\newcommand{\ket}[1]{|#1\rangle}

\usepackage{graphicx}
\usepackage{setspace}
\usepackage{verbatim}
\usepackage{subfig}
\usepackage{xcolor}

\numberwithin{equation}{section}

\DeclareRobustCommand\openone{\leavevmode\hbox{\small1\normalsize\kern-.33em1}}

\newcommand{\be}{\begin{equation}}
	\newcommand{\ee}{\end{equation}}
\newcommand{\bea}{\begin{eqnarray}}
	\newcommand{\eea}{\end{eqnarray}}
\newcommand{\beas}{\begin{eqnarray*}}
	\newcommand{\eeas}{\end{eqnarray*}}

\DeclareFontFamily{U}{mathx}{\hyphenchar\font45}
\DeclareFontShape{U}{mathx}{m}{n}{<-> mathx10}{}
\DeclareSymbolFont{mathx}{U}{mathx}{m}{n}
\DeclareMathAccent{\widebar}{0}{mathx}{"73}
\DeclareMathAccent{\widehat}{0}{mathx}{"70}
\DeclareMathAccent{\widecheck}{0}{mathx}{"71}

\usepackage{cite}
\usepackage{amsmath,amssymb,amsfonts}
\usepackage{algorithmic}
\usepackage{graphicx}
\usepackage{textcomp}
\usepackage{xcolor}
\usepackage{hyperref}
\def\BibTeX{{\rm B\kern-.05em{\sc i\kern-.025em b}\kern-.08em
    T\kern-.1667em\lower.7ex\hbox{E}\kern-.125emX}}
\begin{document}

\title{Information-Theoretic Limits of Quantum Learning  via Data Compression
\thanks{This project was funded by European Union’s Horizon 2020 Research and Innovation Program through H2020-FETOPEN-PHOQUSING (Grant Number: 899544);
Quantum Information Center Sorbonne (QICS); and
Sandoz Family Foundation-Monique de Meuron Program for Academic Promotion.}
}

\author{\IEEEauthorblockN{Armando Angrisani}
\IEEEauthorblockA{
\textit{LIP6, CNRS}, \\ \textit{Sorbonne Université} \\
Paris, France
 \\ \emph{and} \\
 \textit{Institute of Physics} \\
\textit{Ecole Polytechnique Fédérale de Lausanne}\\
Lausanne, Switzerland 
\\
\href{mailto:armando.angrisani@epfl.ch}{armando.angrisani@epfl.ch}}
\and
\IEEEauthorblockN{Brian Coyle}
\IEEEauthorblockA{\textit{QC Ware} \\
{Palo Alto, USA}\\
Paris, France \\
 \emph{and} \\
\textit{School of Informatics }\\
\textit{University of Edinburgh},\\ Edinburgh, United Kingdom }
\and
\IEEEauthorblockN{Elham Kashefi}
\IEEEauthorblockA{\textit{School of Informatics }\\
\textit{University of Edinburgh},\\ Edinburgh, United Kingdom \\ \emph{and} \\
\textit{LIP6, CNRS}, \\ \textit{Sorbonne Université} \\
Paris, France}
}

\maketitle

\begin{abstract}
Understanding the power of quantum data in machine learning is central to many proposed applications of quantum technologies. While access to quantum data can offer exponential advantages for carefully designed learning tasks and often under  strong assumptions on the data distribution, it remains an open question whether such advantages persist in less structured settings and under more realistic, naturally occurring distributions. Motivated by these practical concerns, we introduce a systematic framework based on quantum lossy data compression to bound the power of quantum data in the context of probably approximately correct (PAC) learning.
Specifically, we provide lower bounds on the sample complexity of quantum learners for arbitrary functions when data is drawn from Zipf's distribution, a widely used model for the empirical distributions of real-world data. We also establish lower bounds on the size of quantum input data required to learn linear functions, thereby proving the optimality of previous positive results.
Beyond learning theory, we show that our framework has applications in secure delegated quantum computation within the measurement-based quantum computation (MBQC) model. In particular, we constrain the amount of private information the server can infer, strengthening the security guarantees of the delegation protocol proposed in (Mantri et al., PRX, 2017) \cite{mantri_flow_2017}.
\end{abstract}


% \section*{Preliminaries}

\section*{Introduction} \label{sec:intro}

In recent years machine learning algorithms found impressive applications in several domains, ranging from automated driving, speech recognition, fraud detection and many others. Among the theoretical models proposed for the analysis of such algorithms, the \emph{probably approximately correct} (PAC) learning is undoubtedly the most successful. Introduced in $1984$ by Leslie Valiant in the seminal work ``A theory of the learnable'' \cite{valiant_theory_1984}, this model is characterized by a twofold notion of error: given a target function $f$, the learner is required to output a \emph{hypothesis} function $h$ which is close to $f$ up to a tolerance $\epsilon$, with probability $1- \delta$.

The origins of quantum computation also find roots in the $1980$s, when Richard Feynman and others proposed harnessing quantum physics to create a novel model of computation \cite{benioff_computer_1980, feynman_simulating_1982,preskill_quantum_2021}. Formalising this, we have now substantial theories built around this idea and how information can be transferred and processed in a quantum manner. This leads to \emph{quantum} information theory and \emph{quantum} computation respectively. In the modern day, Feynman's vision is coming ever closer to realisation, in particular with the rapid development of small quantum computers. These devices with $\sim 80$ - $200$ qubits are dubbed noisy intermediate scale quantum (NISQ)~\cite{preskill_quantum_2018} computers. On the quantum communication side, we also have proposals for a `quantum internet', enabled by the secure transfer of quantum information across large distances~\cite{wehner_quantum_2018}. This era poses two interesting challenges. On the one hand, due the low number of qubits and the experimental noise, these devices cannot perform many of the algorithms (or protocols on the communication side) thought to demonstrate exponential speedups over classical algorithms.
Thus, the quest for practical applications gained momentum over the last years, especially in the field of quantum machine learning~\cite{biamonte_quantum_2017, ciliberto_quantum_2018, adcock_advances_2015, benedetti_parameterized_2019, bharti_noisy_2021, arunachalam_guest_2017}.
Broadly, two key motivations drive this field: the potential for quantum algorithms to accelerate classical learning tasks~\cite{harrow2009quantum, kerenidis2017quantum, kerenidis2019q}, and the opportunity to exploit genuinely quantum data using quantum processors~\cite{huang2021power, chen2022exponential, hadiashar2023optimal, anshu2024survey}. The latter perspective is particularly compelling, as it leverages the unique informational properties of quantum systems. In this work, we explore this direction through the lens of quantum information theory. Specifically, we use lossy quantum data compression~\cite{datta_one-shot_2013, wilde_quantum_2013} -- an extension of the classical framework of one-shot rate-distortion theory~\cite{kostina2012fixed} --  as a tool to quantify the power of quantum data in machine learning~\cite{hadiashar2023optimal, caro2024information}.
Finally, we also show how this framework can also be harnessed for the design of privacy-preserving protocols for delegated quantum computation, an increasingly relevant setting as quantum hardware becomes remotely accessible via cloud-based platforms~\cite{larose_overview_2019}.

\begin{figure*}[!htb]
    \centering
    \includegraphics[width=0.7\textwidth]{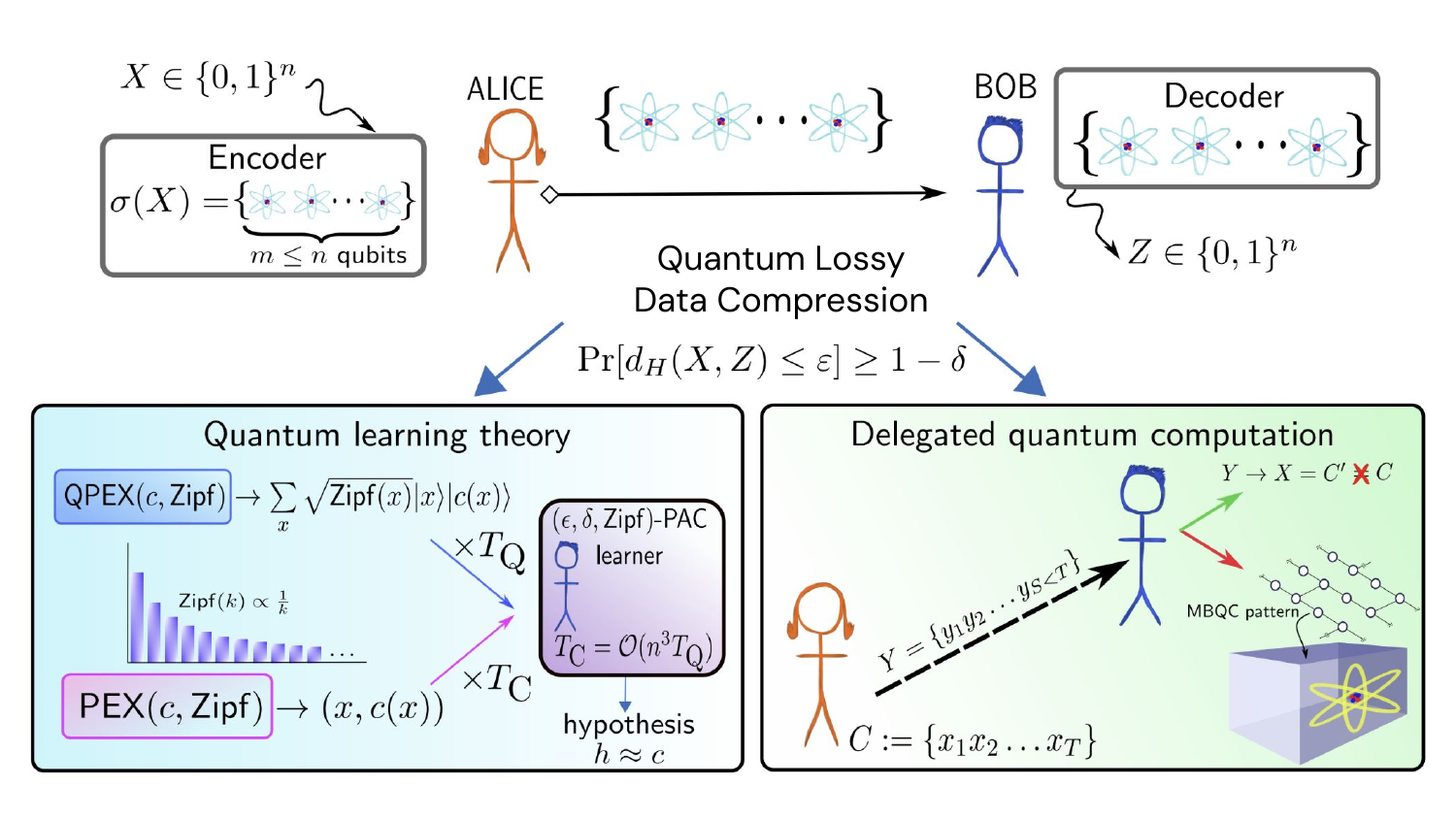}
    \caption{\small\textbf{Quantum lossy data compression and applications.} An overview of the quantum lossy data compression framework and the applications we study in this work. The PAC Nayak bound we derive places bounds on the receiver in a source coding protocol decoding the correct binary string, encoded in less qubits than information-theoretically required to perform a perfect decoding. The first application is in the quantum sample complexity of supervised learning under the $\mathsf{Zipf}$ distribution. Here, PAC source coding can be used to derive a sample complexity for quantum learners which may only be polynomially better than using classical samples. The second is in proving bounds on the ability of an adversary in a delegated quantum computation protocol (in the measurement-based framework of quantum computing (MBQC)) to guess the client's chosen computation.}
    \label{fig:overview_qsc}
\end{figure*}

\subsection{Our contribution} %In this work, we demonstrate that the \emph{probably approximately correct} approach can be employed in quantum information theory, providing applications in quantum learning theory and classically-driven delegated quantum computation. 
In this work, we establish rigorous bounds on the role of data in quantum machine learning, with applications to several fundamental learning problems. Our results also reveal a surprising connection to the security of delegated quantum computation. %To achieve this, we extend the framework of quantum lossy data compression introduced in Ref.~\cite{}.

First, we show that the quantum sample complexity and the classical sample complexity of the class of arbitrary functions under Zipf's distribution are equal, up to poly-logarithmic factors. 
Second, we lower bound the required size of quantum data required for approximately learning linear functions, thereby establishing the optimality of previous approaches requiring $\mathcal{O}(n)$ input qubits~\cite{caro2020quantum}.
Finally, we strengthen the security guarantee of the delegation protocol introduced in Ref.~\cite{mantri_flow_2017}, bounding the probability that server reconstruct an approximated version of the target computation.
%Whereas we stated our bounds with respect to quantum states, the classical versions of our results can be recovered as special cases. Indeed, we believe that the application of our approach to classical computer science can be of independent interest.

Our results suggest that lossy data compression provides a powerful framework for unifying concepts from information theory, machine learning, and delegated computation, with the potential to drive further advances in all three areas.

\section{Preliminaries}
%In order to investigate the power of quantum data in machine learning, 
%We first recall some preliminary results. %on quantum source coding and quantum lossy data compression.
\noindent \textbf{Quantum source coding.} The term \emph{source coding} refers to the process of encoding information produced by a given source in a way that it may be later decoded. The initial result of this topic, the \emph{source coding theorem}, was derived in the seminal work of Shannon~\cite{shannon_mathematical_1948}, which describes how many bits are required to encode independent and identically distributed random variables, without loss of information. One way to formalize this scenario is to consider a communication task, between two parties, Alice (A) and Bob (B). One generalisation of this theorem to the quantum world, is via Schumacher's theorem~\cite{schumacher_quantum_1995, jozsa_new_1994} which describes the number of \emph{qubits} required to compress an $n$-fold tensor product of a quantum state, $\rho^{\otimes n}_A$, sent between Alice and Bob. In both cases, this number is directly related to the \emph{entropy} of the random variable, or the quantum state. In the quantum scenario, the relevant quantity is the von Neumann entropy:
\begin{equation}\label{eqn:von_neumann_entropoy}
    S(A) := S(\rho_A) := - \tr\left(\rho_A \log\rho_A \right)
\end{equation}
%
%This quantity generalises its classical counterpart, the Shannon entropy, which is the relevant object in the classical source coding theorem. 
From the entropy, one may define the \emph{mutual information}, given in the quantum case by:
\begin{equation}
    I(A;B) := S(A) + S(B) - S(AB)
\end{equation}
We use the notation $S(AB)$ to mean the entropy of the state, $\rho_{AB}$, shared between Alice and Bob, and $\rho_A = \tr_B(\rho_{AB})$ is the reduced density matrix of this state with entropy, $S(A)$, and likewise for Bob. 

However, all of the above applies to \emph{purely} quantum sources, which is actually slightly too general for our purposes. Instead, we only require the source of Alice to quantumly encode \emph{classical} information. 
The setting is the following. Alice samples an $n$-bit random string, $X \xleftarrow{}  \{0, 1\}^n$, from some distribution, $p:\{0, 1\}^n \rightarrow [0, 1]$. She will then encode $X$ into a quantum state, $\sigma(X) \in \mathbb{C}^{d \times d}$ and send it to Bob. This defines an ensemble, $\Sigma := \{\sigma(X), p(X)\}$. Finally, Bob performs a positive operator value measurement (POVM), $\mathbf{E} := \{E_m\}_m$ on $\sigma$, and attempt to extract the encoded information. In reality, Bob will observe a string, $Z \in \{0, 1\}^n$, and we wish to now bound the probability that he was successful, i.e. that $X=Z$.

First, we explore the minimum number of qubits Alice need to send to Bob to be successful in extracting the correct message? The answer comes from the Holevo's bound~\cite{holevo_bounds_1973} which states that the `\emph{accessible information}' of the ensemble, $\Sigma$, defined as:
\begin{equation} \label{eqn:accessible_information}
    I_{\text{Acc}}(\Sigma) = \max_{\mathbf{E}} I(X; Z)
\end{equation}
is bounded by the `Holevo quantity' of the ensemble, $\chi(\Sigma)$:
\begin{equation} \label{eqn:holevo_bound}
    I_{\text{Acc}}(\Sigma) \leq \chi(\Sigma) \leq \log d
\end{equation}
From the Holevo's bound, we know that at least $n$ qubits are required to send a classical message of length $n$. However, what if we do not have access to $n$ qubits, but instead only $m < n$. Now, we have no chance to output the correct string with certainty, but we can aim to bound the probability of decoding the correct string. In this setting, we have the following result known as `Nayak's bound'~\cite{nayak_optimal_1999}:

\begin{lemma}[Nayak's bound] \label{lem:nayaks_bound}
    If $X$ is a $n$-bit binary string, we send it using $m$ qubits, and decode it via some mechanism back to an $n$-bit string $Z$, then our probability of correct decoding is given by:
    \begin{equation} \label{eqn:nayaks_bound}
        \Pr[X=Z] \leq \frac{2^m}{2^n} = 2^{m-n}.
    \end{equation}
\end{lemma}
%
%Which intuitively states that our chance of correct decoding the string decreases exponentially with the difference between $m$ and $n$.

\noindent\textbf{Quantum lossy data compression.}
%Now that we have discussed the required preliminaries, let us move to our results. To begin, let us recall how we discussed above that Nayak's bound allows the string-decoding protocol to fail with some probability $0 \leq \delta <1$. However, 
For our purposes, it is convenient to consider the relaxed framework of \emph{quantum lossy data compression}~\cite{datta_one-shot_2013, datta_quantum--classical_2013},  which involves a second error parameter, $0\leq \varepsilon <1$. We say that Alice and Bob succeed their task with probability $1-\delta$ up to an approximation error $\epsilon$ if
\[
\Pr[d_{H}(X,Z) \leq \varepsilon] \geq 1- \delta,
\]
where $d_H$ is the Hamming distance between two strings of equal length, which is the number of positions at which the corresponding values are different.

\section*{Results}

In this context, we can prove the following generalization of the Nayak's bound (Lemma~\ref{lem:nayaks_bound}).

\begin{lemma}[PAC Nayak bound]
\label{lem:PACNayak}
If $X$ is an $n$-bit binary string, we send it using $m$ qubits, and decode it via some mechanism
back to an $n$-bit string $Z$ , then our probability of correct decoding up to an error $\beta n \geq 0$ in Hamming distance  is given by
\begin{equation} \label{eqn:pac_nayak_bound}
    \Pr[d_H(X,Z)\leq \beta n] \leq \frac{2^m}{2^{n(1-\beta-H(\beta))}}.
\end{equation}

\end{lemma}

\begin{proof}[Proof of Lemma~\ref{lem:PACNayak}]
The proof is straightforward, we begin by explicitly expanding the left hand side of~\eqref{eqn:pac_nayak_bound}:
\begin{align*}
 &\Pr[d_H(X,Z)\leq \beta n] \\= &\Pr[\exists i_1\dots i_{ \lceil (1-\beta)n \rceil}: X_{i_1}\dots X_{i_{ \lceil (1-\beta)n \rceil }} = Z_{i_1}\dots Z_{i_{\lceil (1-\beta)n\rceil }}] \\ &\leq  \binom{n}{\lfloor \beta n \rfloor} \Pr[ X_{i_1}\dots X_{i_{\lceil (1-\beta) n\rceil }} = Z_{i_1}\dots Z_{i_{\lceil (1-\beta) n \rceil }}] \label{eqn:pac_nayak_proof_2}\\
 &\leq 2^{H(\beta) n}\frac{2^m}{2^{n(1-\beta)}} = \frac{2^m}{2^{n(1-\beta-H(\beta))}},%\label{eqn:pac_nayak_proof_3} 
\end{align*}
% where the second inequality follows from the union bound and the last inequality from Lemma \ref{lem:comb}.
We use the union bound for the second inequality, and \eqref{eqn:binomial_coefficient_bound}.

\end{proof}
Now, let us move from a success probability to a sample complexity result. In a PAC learning scenario, it is natural to assume that a classical message is transmitted through $m$ identical copies of a quantum state. This hypothesis stems from the definition of quantum sample reviewed in \eqref{eq:sample}. Now, a natural question is to ask how many copies of such states are required in order to learn $X$ up to an error $\epsilon$ in Hamming distance with probability $1-\delta$. We can derive the following lower bound from Holevo's bound.

\begin{lemma}[Learning a string with quantum data]
\label{lem:PACHolevo}
Let $\epsilon \in [0,1/2)$.
Assume $X$ is an $n$-bit binary string sampled with probability $p$, we send it using $m$ copies of a quantum state $\rho \in \mathbb C^{2^\ell \times 2^\ell}$ and decode it via some mechanism back to an $n$-bit string $Z$. Let $d_H(X,Z) \leq \epsilon n$ with probability $1-\delta$. Then,
\[
m \geq \frac{(1-\delta)(1-H(\epsilon))n - H(\delta)}{\ell}
\]
\end{lemma}

\begin{proof}
By Lemma~\ref{lem:copies}, we have that $I_{\text{Acc}}(\Sigma = \{\rho^{\otimes m}, p\}) \leq m I_{\text{Acc}}(\Sigma' = \{\rho, p\})$.
% I_{acc}(\rho^{\otimes m}, p) \leq m I_{acc}(\rho, p)$. 
Moreover, we also have that $I_{\text{Acc}}(\Sigma' = \{\rho, p\})  \leq \ell$ by~\eqref{eqn:holevo_bound}.
By applying Lemma~\ref{lem:qpsc} and taking $\alpha:= m \ell$, we have that

\begin{align}
    (1 - \delta)(1 - H(\epsilon))n - H(\delta) &\leq  m \ell.\\
    \implies m &\geq \frac{(1 - \delta)(1 - H(\epsilon))n - H(\delta)}{\ell}
\end{align}
which completes the proof.
\end{proof}

Using the bounds above, we demonstrate two use cases in apparently distinct application areas, namely quantum learning theory and secure delated quantum computation. 
%The first is an example in quantum learning theory, in which we prove a lower bound on the sample complexity of supervised learning algoritheorems relative to a specific distribution.
%The second is in the area of \emph{delegated} quantum computation, where we use the bound to analyse and refine the security of the protocol of~\cite{mantri_flow_2017}. This protocol describes a method for a resource-limited `client' to delegate a quantum computation to a powerful quantum `server' in a manner to provide security guarantees to the data and information of the delegating client. 

\subsection{Learning with quantum data under the Zipf's law}
\label{sub:qlt}
Before presenting our result, let us first briefly formalise our discussions of quantum learning theory. %, and how the above results can be useful there. As mentioned above, 
One of the most popular formalisms for a theory of learning is Probably Approximately Correct (PAC) learning. In PAC learning, one considers a \emph{concept class}, usually Boolean functions, $\mathcal{C} \subseteq \mathcal{F}_n := \{ f | f : \{0,1\}^n \xrightarrow{} \{0,1\} \}$. For a given concept (given Boolean function to be learned), $c\in \mathcal{C}$, the goal of a PAC learner is to output a `hypothesis', $h$, which is `probably approximately' correct. In other words, the learner can output a hypothesis which agrees with the output of chosen function almost always. In order to be a `PAC learner for a concept class', the learner must be able to do this \emph{for all} $c \in \mathcal{C}$. In this process, learners are given access to an oracle\footnote{Formally, this is a \emph{random example oracle} for the concept class under the distribution, $O = \mathsf{PEX}(c, D)$. The type of oracle the learner has access to may make the learning problem more or less difficult.}, $O$, which outputs a sample from a distribution, $x \sim D$, along with the corresponding concept evaluation (the label) at that datapoint, $c(x)$. It is also important to distinguish between distribution-dependent learning, where $D$ is fixed, and distribution-free learning, where $D$ is arbitrary.

Now, we can formally define a PAC learner. Given a target concept class $\mathcal{C}$, a distribution $D$ and $\epsilon, \delta \in [0,1]$, for any $c \in \mathcal{C}$ a $(\epsilon, \delta, D)$-PAC learner outputs a hypothesis $h$ such that
\begin{equation}
\label{eq:pac}
    \ell(c,h):= \Pr_{x\sim D}[c(x)\neq h(x)] \leq \epsilon
\end{equation}
with probability at least $1-\delta$.

In \emph{quantum} PAC learning, we consider an alternative oracle $O'$\footnote{A \emph{quantum} random example oracle, $O' = \mathsf{QPEX}(c, D)$.}, as defined in~\cite{bshouty_learning_1995}, which outputs a quantum state that encodes both the concept class and the distribution, that is
\begin{equation}
\label{eq:sample}
    \sum_{x\in \{0,1\}^n} {\sqrt{D(x)}}\ket{x}\ket{c(x)}.
\end{equation}

%A quantum oracle can easily simulate a classical one: it suffices to measure the first register of the state above in the computational basis.
A key question is whether quantum learners may be able to learn concepts with a lower sample complexity than is possible classically. Early results in this direction were both positive and negative, with the distribution from which the examples are sampled being a crucial ingredient. For example, it was shown in~\cite{bshouty_learning_1995} that exponential advantages for PAC learners were possible under the \emph{uniform} distribution. On the other hand, in the distribution-free case, there is only a marginal improvement that quantum samples can hope to provide~\cite{arunachalam_optimal_2018}. For our purposes, we analyse the sample complexity of quantum PAC learners, and use our tools to derive a lower bound on the learning problem, relative to another specific distribution (not simply the uniform one). The distribution in question is the \emph{Zipf} distribution (which is a long-tailed distribution relevant in many practical scenarios, see~\cite{zhu_capturing_2014, zhang_understanding_2017, wang_learning_2017} for example) over $\{0, 1\}^n$, defined as follows:
\begin{equation} \label{eqn:zipf_distribution}
    \mathsf{Zipf}(k) := \frac{1}{k H_N},~ H_K := \sum\limits_{k=1}^K \frac{1}{k}\in[\log k, \log k +1]
\end{equation}

In the following, we allow the target function to be completely arbitrary, i.e. we study the learnability of the concept class $\mathcal{F}_n$. Though this framework may look too simplistic, it captures prediction problems in which the domain  doesn't have any underlying structure. Examples of such structure could be a notion of distance on the domain. It has been thoroughly studied in \cite{feldman}, especially in relation to long-tailed distributions.

By applying~Lemma~\ref{lem:PACHolevo}, we can prove the following lower bound on the sample complexity of a quantum PAC learner for $\mathcal{F}_n$ relative to the Zipf distribution. 

\begin{theorem}
\label{theorem:zipf}
Let $N=2^n$ and $\beta \in [0,1/2]$.
For every $\epsilon \in [0,\frac{\beta}{n+1})$ and $\delta \in [0,1]$, every $(\epsilon,\delta, \mathsf{Zipf})$-PAC  quantum learner for $\mathcal{F}_n$ has sample complexity:
\begin{equation}
    \Omega\left(\frac{(1-\delta)(1-H(\beta))N - H(\delta)}{n}\right).
\end{equation}
\end{theorem}

\begin{proof}
Firstly, let $h$ be the hypothesis produced by a quantum learner upon receiving $m$ quantum samples as in \eqref{eq:sample}, such that $\ell(f,h)\leq \epsilon$ with probability at least $1-\delta$.

In this proof, we represent the target function and the hypothesis in $N$-bit strings: $X=f(0)f(1)\dots f(N)$ and $Z=h(0)h(1)\dots h(N)$.
Let $I$ be the subset of the indices where $X$ and $Z$ differ, i.e. $I:=\{i\in [N]:f(i)\neq h(i)\}$.
The condition $\ell(f,h)\leq \epsilon$ can be restated as:
\[
\mathsf{Zipf}(I):= \sum_{i\in I}\mathsf{Zipf}(i) \leq \epsilon < \frac{\beta}{n+1}.
\]

Now we lower bound the quantity above.
\[
\sum_{i\in I}\mathsf{Zipf}(i) \geq |I| \min_{i \in [N]}\mathsf{Zipf}(i)\geq \frac{|I|}{N(n+1)}.
\]
Combining these two inequalities, we get that $|I|<\beta N < \frac{N}{2}$. Furthermore, we can observe that that $d_H(X,Z) = |I| < \beta N$.

Finally, by~Lemma~\ref{lem:PACHolevo}, since quantum samples exist in $n$-qubit states, the sample complexity of  quantum learner can be lower bounded as follows:
\[
    m \geq \frac{(1-\delta)(1-H(\beta))N - H(\delta)}{n},
\]
which completes the proof.
\end{proof}

%By setting an upper bound on the approximation error, we retrieve the following corollary.

%\begin{cor}
%Let $N = 2^n$. For every $\epsilon = O(1/n)$ and $\delta \in [0,1]$, every $(\epsilon,\delta, \mathsf{Zipf})$-PAC  quantum learner for $\mathcal{F}_n$ has sample complexity $\Omega\left(N/{ n}\right)$.
%\end{cor}
In particular, this result implies that we need at least ${O}(N/n)$ quantum samples to learn an arbitrary function under the \textsf{Zipf} distribution with approximation error $\epsilon = {O}(1/n)$.

Finally, for any distribution $D$ over $\{0,1\}^n$, we exhibit a $(\epsilon,\delta, D)$-PAC classical learner for $\mathcal{F}_n$ with sample complexity $O\left(\frac{N}{\epsilon}\log\frac{N}{\delta}\right)$. 
Informally, such learner memorizes the labels of the training examples, and assigns a random label to the remaining instances.

\begin{theorem}
\label{theorem:z1}
Let $N=2^n$ and let $D$ be a distribution over $[N]$.
There exists a classical $(\epsilon, \delta, D)$-PAC learner with sample complexity $O\left(\frac{N}{\epsilon}\log\frac{N}{\delta}\right)$. 
\end{theorem}

\begin{proof}
 We partition $X= [N]$ in two sets:
\begin{itemize}
    \item $X_1 := \{i \in [N]: D(i) > \epsilon/N\}$.
    \item $X_2 := \{i \in [N]: D(i) \leq \epsilon/N\}$.
\end{itemize}

Observe that if $\forall x \in X_1 : h(x) = f(x) $, then the generalization error is at most $\epsilon$.

Fix $m := { \frac{N}{\epsilon}\log\frac{N}{\delta}}$.
Let $S = \{({i_1},f({i_1})), \ldots, ({i_m},f({i_m})) $ be the examples.

The algorithm $\mathcal{A}$ works as follows:
\begin{enumerate}
    \item For each point $i$ in the sample $S$, memorize its label $f(i)$.
    \item Assign a random label to the points that aren't in the sample $S$.
\end{enumerate}

It suffices to prove that, with probability $1 - \delta$, the sample covers all the points in $X_1$, i.e. $X_1\subseteq \{{i_1},\ldots, {i_m}\}$. 
Consider the probability of not obtaining a point $i \in X_1$ after $m$ queries, which is given by:
\[
(1- D(i))^m \leq \left( 1- \frac{\epsilon}{N} \right)^{ \frac{N}{\epsilon}\log\frac{N}{\delta}} < e^{\log\frac{\delta}{N}} = \frac{\delta}{N},
\]
where we employed the inequality $(1-\frac{1}{y})^y \leq \frac{1}{e}$ that is true for all $y \geq 1$.
By the union bound, the probability that a point, $x$, in $X_1$ does not appear in the sample is given by
\begin{align*}
 &\Pr[\exists x \in X_1 : x_1 \not\in \{{i_1},\ldots, {i_m}\}] \\\leq &\sum_{x \in X_1} \Pr[x \not\in \{{i_1},\ldots, {i_m}\}] \leq N \frac{\delta}{N} = \delta.   
\end{align*}

Thus,  with probability at least $1-\delta$ every point in $X_1$ is contained the sample.
\end{proof}

The former result implies that we can learn an arbitrary function under an arbitrary distribution with approximation error $\epsilon = {O}(1/n)$ using at most ${O}(N n^2)$ classical samples.

Putting these together, we can observe that if we receive $T_{\text{Q}}$ quantum samples as in \eqref{eq:sample} according to the $\mathsf{Zipf}$ distribution, we can achieve the same learning error with $T_{\text{C}} = n^3T_{\text{Q}}$ classical samples, i.e. the sample complexity decreases by at most of a factor ${{O}}(n^3)$.

\subsection{Learning linear functions with quantum data}
Our next result concerns \emph{linear} functions, also known as parity functions, which have been extensively studied in learning theory and have found several applications in cryptography. 
A major result in quantum learning theory establishes an unbounded separation for learning linear functions with quantum examples compared to classical examples: a single quantum example suffice to learn linear function, compared to  $\Theta(n)$ classical examples~\cite{arunachalam_guest_2017, caro2020quantum}.
We observe that a quantum examples consists in $n+1$ qubits. A natural question is whether there exists other kind of quantum resources beyond quantum examples that allow to further reduce this qubit count. 
Here, we provide a negative answer, demonstrating that $\Omega(n)$ qubits are necessary to $(\epsilon,\delta)$-PAC learn linear functions.

In order to state this result, we first prove the following Lemma.
{
%\color{blue}
\begin{lemma}[Approximate Recovery of Linear Functions]
Let $\eta > 0$ be a constant. Consider a distribution $\mathcal{D}$ over $\{0,1\}^n$ such that for all $y,z \in \{0,1\}^n$ satisfying $d_H(y,z)=1$, it holds that
\[
    \Pr_{x\sim \mathcal{D}}[x=y] \leq \eta \cdot \Pr_{x\sim \mathcal{D}}[x=z].
\]
Let $z \in \{0,1\}^n$ be unknown and define the linear function $c_z(x) = x \cdot z \mod 2$. Suppose $h : \{0,1\}^n \rightarrow \{0,1\}$ satisfies
\[
    \Pr_{x\sim \mathcal{D}} [h(x) \neq c_z(x)] \leq \epsilon.
\]
Then there exists a classical algorithm that, using $N$ samples from $\mathcal{D}$, outputs a string $y \in \{0,1\}^n$ such that with high probability
\[
    d_H(y,z) \leq (1+\eta)\epsilon n + O\left(\sqrt{\frac{n \log(n/\delta)}{N}}\right).
\]
In particular, for $N = O\left(\epsilon^{-2} \log(n/\delta)\right)$, we have with probability at least $1 - \delta$ that
\[
    d_H(y,z) \leq 2(1+\eta)\epsilon n.
\]
\end{lemma}

\begin{proof}
For each $i \in [n]$, let $e_i \in \{0,1\}^n$ be the vector with a 1 in coordinate $i$ and zeros elsewhere. Then for any $x \in \{0,1\}^n$,
\begin{align*}
    c_z(x) \oplus c_z(x \oplus e_i) 
    &= (x \cdot z \mod 2) \oplus ((x \oplus e_i) \cdot z \mod 2) \\
    &= (x \cdot z) \oplus ((x \cdot z) \oplus z_i) \mod 2 \\
    &= z_i.
\end{align*}
Define the estimator:
\[
    \tilde{z}_i(x) := h(x) \oplus h(x \oplus e_i).
\]
This is intended to approximate $z_i$. Now we bound the error of $\tilde{z}_i(x)$:
\begin{align}
        &\Pr_{x \sim \mathcal{D}}[\tilde{z}_i(x) \neq z_i] 
        \\\leq &\Pr_{x \sim \mathcal{D}}[h(x) \neq c_z(x)] + \Pr_{x \sim \mathcal{D}}[h(x \oplus e_i) \neq c_z(x \oplus e_i)].
\end{align}

We already have
\[
    \Pr_{x \sim \mathcal{D}}[h(x) \neq c_z(x)] \leq \epsilon.
\]
To bound the second term, note:
\begin{align*}
    &\Pr_{x \sim \mathcal{D}}[h(x \oplus e_i) \neq c_z(x \oplus e_i)] 
    = \sum_{x \in \{0,1\}^n} \Pr[x] \cdot \mathbf{1}_{h(x \oplus e_i) \neq c_z(x \oplus e_i)} \\
    = &\sum_{y \in \{0,1\}^n} \Pr[x = y \oplus e_i] \cdot \mathbf{1}_{h(y) \neq c_z(y)} \\
    \leq &\sum_{y \in \{0,1\}^n} \eta \cdot \Pr[x = y] \cdot \mathbf{1}_{h(y) \neq c_z(y)} 
    = \eta \cdot \Pr_{x \sim \mathcal{D}}[h(x) \neq c_z(x)] 
    \leq \eta \epsilon.
\end{align*}
So,
\[
    \Pr_{x \sim \mathcal{D}}[\tilde{z}_i(x) \neq z_i] \leq (1+\eta)\epsilon.
\]

Now take $N$ independent samples $x^{(1)}, \ldots, x^{(N)} \sim \mathcal{D}$ and define
\[
    \widehat{\mu}_i := \frac{1}{N} \sum_{j=1}^N \tilde{z}_i(x^{(j)}).
\]
Then by Hoeffding’s inequality, for any $t > 0$,
\[
    \Pr\left[\left| \widehat{\mu}_i - \mathbb{E}[\tilde{z}_i(x)] \right| > t\right] \leq 2\exp(-2t^2 N).
\]
By a union bound over all $i \in [n]$,
\[
    \Pr\left[\exists i \in [n] : \left| \widehat{\mu}_i - \mathbb{E}[\tilde{z}_i(x)] \right| > t\right] \leq 2n \exp(-2t^2 N).
\]

Now observe that
\[
    |\mathbb{E}[\tilde{z}_i(x)] - z_i| = \left| \Pr[\tilde{z}_i = 1] - z_i \right| \leq (1+\eta)\epsilon,
\]
since $z_i \in \{0,1\}$. Therefore, rounding $\widehat{\mu}_i$ to the nearest bit recovers $z_i$ provided the empirical mean is within $(1/2 - (1+\eta)\epsilon)$ of the true value.

To ensure this, set $t = 1/4$ and choose $N = O(\log(n/\delta))$. With probability at least $1 - \delta$, all $\widehat{\mu}_i$ are within $t$ of $\mathbb{E}[\tilde{z}_i(x)]$, and so the rounding procedure:
\[
    y_i := \begin{cases}
        1 & \text{if } \widehat{\mu}_i > 1/2, \\
        0 & \text{otherwise},
    \end{cases}
\]
satisfies
\[
    d_H(y,z) \leq (1+\eta)\epsilon n + O\left(\sqrt{\frac{n \log(n/\delta)}{N}}\right)
\]
by a Chernoff bound.
\end{proof}

}
%\begin{theorem}
%Let $\rho_x$ be an $m$-qubit quantum state encoding the parity function $c_x$. Assume there's an algorithm that PAC learns $x$ from $\rho_x$. Then 
%\begin{align}
%    m \geq ?? n.
%\end{align}
%\end{theorem}

 We now prove the main result of this section.
{%\color{blue}
\begin{theorem}
Let $\{\rho_x : x\in\{0,1\}^n\}$ be an $m$‑qubit encoding of the linear function $c_x:\{0,1\}^n\to\{0,1\}$, and suppose there is a (possibly randomized) decoding procedure which, on input $\rho_x$, outputs an $n$‑bit string $Z$ satisfying
\[
  \Pr\bigl[d_H(x,Z)\le\beta n\bigr]\;\ge\;1-\delta
  \quad\text{for some }\beta<\tfrac12,\ \delta<1.
\]
Then 
\[
  m \;\ge\; n\bigl(1-\beta-H(\beta)\bigr)\;-\;\log\frac1{1-\delta}
  \;=\;\Omega(n).
\]
In particular, for any fixed success‐probability $1-\delta=\Omega(1)$ and error‐rate $\beta<1/2$, one has
\[
  m \;\ge\; \bigl(1-\beta-H(\beta)\bigr)\,n \;-\; O(1)
  = \Omega(n).
\]
\end{theorem}

\begin{proof}
By hypothesis, the decoding succeeds with
\[
  \Pr_{x,Z}\bigl[d_H(x,Z)\le\beta n\bigr]\;\ge\;1-\delta.
\]
But Lemma~\ref{lem:PACNayak} (the PAC Nayak's bound) says that for \emph{any}
$m$‑qubit encoding plus decoding,
\[
  \Pr_{x,Z}\bigl[d_H(x,Z)\le\beta n\bigr]
  \;\le\;
  \frac{2^m}{2^{\,n\bigl(1-\beta-H(\beta)\bigr)}}.
\]
Combine the two:
\[
  1-\delta\;\le\;\frac{2^m}{2^{\,n\bigl(1-\beta-H(\beta)\bigr)}}
  \;\Longrightarrow\;
  2^m\;\ge\;(1-\delta)\;2^{\,n\bigl(1-\beta-H(\beta)\bigr)}.
\]
Taking binary logarithms,
\begin{align*}
      m &\;\ge\; n\bigl(1-\beta-H(\beta)\bigr)\;+\;\log(1-\delta)
  \\&\;=\; n\bigl(1-\beta-H(\beta)\bigr)\;-\;\log\!\tfrac1{1-\delta}.
\end{align*}

Since $\beta<1/2$ implies $1-\beta-H(\beta)>0$, this is $\Omega(n)$ as claimed.
\end{proof}
}
Together with Ref.~\cite{caro2020quantum}, the results provided in this section imply that $\Theta(n)$ qubits are necessary and sufficient to PAC learn linear functions from quantum data.
\subsection{Classically-driven delegated quantum computation}
\label{sub:flow}
In the previous sections, we demonstrated an application of our bounds in learning theory. Interestingly, the same mathematical tools can be applied in a completely different context, namely in the delegation of (quantum) computation. In this scenario, Alice takes the role of a `client', who wishes to perform a quantum computation. However, she does not have the resources to to implement the computation herself, and therefore must delegate the problem to be solved to a `server', Bob. However, for any number of reasons, the client may not trust the server, if even only to preserve their own privacy. This situation is extremely relevant in the current NISQ era, in which quantum computers do exist solely in the `quantum cloud'~\cite{larose_overview_2019} and will likely also be the norm for the near and far future, even perhaps when fully fault tolerant quantum computers are available. Delegated quantum computation has arisen to deal with this possibility and ensure computations are performed correctly and securely, with blind quantum computation being one of the primary techniques~\cite{childs_secure_2005, broadbent_universal_2009, fitzsimons_unconditionally_2017}.

However, typically blind quantum computation requires Alice to have some minimal quantum resource, i.e. the ability to prepare single qubit states and send them to Bob. This presents its own issues however several efforts have been made to remove the quantum requirements from Alice~\cite{gheorghiu_computationally-secure_2019, cojocaru_qfactory_2019}. Doing so typically reduces the security of such delegation protocols from information-theoretic to simply computational in nature, based on primitives from post-quantum cryptography. Furthermore, based on assumptions in complexity theory, it is likely that information-theoretic security in this scenario is impossible without such quantum communication~\cite{aaronson_complexity-theoretic_2019}. Nevertheless, it still may not be practical to demand quantum-enabled clients, and so in many cases classical communication to the server may be the best we can do.

In particular, we have in mind a protocol proposed by~\cite{mantri_flow_2017}, which enables the delegation (by Alice) of a quantum computation to a remote server (Bob), using only classical resources (called \emph{classically driven blind quantum computing} (CDBQC)). We remark that here the notion of \emph{blindness} is weaker than in~\cite{childs_secure_2005, broadbent_universal_2009, fitzsimons_unconditionally_2017}. The CDBQC protocol is written in the language of \emph{measurement-based} quantum computation (MBQC). Instead of using the quantum circuit model in which computation is driven by the operation of unitary evolutions on qubits, MBQC starts from a large entangled state (represented as a graph) and is driven by measurement of the qubits in the graph. In order to drive a blind computation, Alice needs to specify a set of measurement angles which describe the basis in which each qubit is measured. 

Now, in the protocol of~\cite{mantri_flow_2017}, she sends Bob $2n$ bits in order to convey this information. However, it is shown in~\cite{mantri_flow_2017} how $T:=3.388n$ bits are required to be send for Bob to have a complete specification of the MBQC computation. These extra bits are required to convey also the `flow information' which is, roughly speaking, the manner in which the desired computation flows through the graph. The work of~\cite{mantri_flow_2017} essentially hides this flow information from the server, so Bob does not know which computation is actually desired by Alice. 
By applying Nayak's bound (Lemma~\ref{lem:nayaks_bound}), the analysis of~\cite{mantri_flow_2017} shows that the probability that the server decodes the entire computation is at most $2^{-1.388n} \simeq 2^{-0.41 T}$.

% In the previous sections, we showed that quantum data turns to be of little interests in some learning problems. The same mathematical tools can be employed in a different context, namely delegated (quantum) computation. As a case study, we refer to the a classical protocol for delegated quantum computation proposed in~\cite{mantri_flow_2017}. They base the security of their protocol on an information-theoretic argument. A quantum computation chosen uniformly at random from a certain set can be represented by  approximately $3.388 N$ bits. Since only $2N$ bits are sent by the client, Nayak's bound (Lemma~\ref{lem:nayaks_bound}) implies that the server can guess the delegated computation with probability at most $2^{-1.388N}$.

Using the tools derived above, we can further refine the claim of~\cite{mantri_flow_2017}. In particular, it is natural to ask whether the server can even guess a significant portion of the computation (rather than the \emph{entire} computation as~\cite{mantri_flow_2017}). In many realistic scenarios, an adversary leaking \emph{any} part of a computation can be a serious threat to the client's privacy.

Now, by applying~Lemma~\ref{lem:qpsc} and~Lemma~\ref{lem:PACNayak} we can derive the following result:

\begin{theorem}
\label{theorem:flow}
Let $\mathcal{C}$ be a family of computations, such that each $C \in \mathcal{C}$ can be encoded in a string of at least $T$ bits. Assume that a client delegates $C \in \mathcal{C}$ to a server using at most $0.591T$ bits of communication. Then the following bounds hold:
\begin{enumerate}
    \item For every $\epsilon \in [0,1]$, the server can guess $C$ up to an error $\epsilon T$ in Hamming distance with probability at most $2^{(-0.41+\epsilon+H(\epsilon))T}$. This gives an exponentially small probability for every $\epsilon < 0.06 $.
    \item Assume that server outputs $C'$ such that $d_H(C,C') \leq  \epsilon T$ with probability $1$. Then $\epsilon > 0.08 $.
\end{enumerate}
\end{theorem}

\begin{proof}
The first claim follows from a direct application of the PAC Nayak bound (Lemma~\ref{lem:PACNayak}). The server receives an input a string $Y$ of less than $0.591T$ bits, and she wants to decode a $T$-bit string $X$, with an approximation error of at most $\epsilon T$ in Hamming distance. Hence, let $Z$ be a $T$-bit string output by the server. The success probability of a correct decoding is then given by:
\[
 P[d_{H}(X,Z)\leq \epsilon T] \leq \frac{2^{0.591T}}{2^{T[1-\epsilon-H(\epsilon)]}}\leq 2^{[-0.41 + \epsilon +H(\epsilon)]T}. 
\]
The second claim in Theorem~\ref{theorem:flow} follows from an application of Lemma~\ref{lem:qpsc}, choosing $\delta = 0$ and combining with~\eqref{eqn:holevo_bound}. The accessible information can be upper bounded by $0.591T$, and thus we have:
\begin{align*}
    (1-H(\epsilon))T \leq 0.591T \qquad
    \implies \epsilon > 0.08.
\end{align*}
 
\end{proof}

% In the previous sections, we showed that quantum data turns to be of little interests in some learning problems. The same mathematical tools can be employed in a different context, namely delegated (quantum) computation. As a case study, we refer to the a classical protocol for delegated quantum computation proposed in~\cite{mantri_flow_2017}. They base the security of their protocol on an information-theoretic argument. A quantum computation chosen uniformly at random from a certain set can be represented by  approximately $3.388 N$ bits. Since only $2N$ bits are sent by the client, Nayak's bound (Lemma~\ref{lem:nayaks_bound}) implies that the server can guess the delegated computation with probability at most $2^{-1.388N}$.

% The former claim can be further refined. It's natural to ask whether the server can guess a significant portion of the computation. In many realistic scenarios, an adversary leaking part of a computation can be a serious threat to the client's privacy. A direct application of the Lemmata introduced in Section \ref{sec:PACSC} gives the following result.
% Interestingly, the PAC source coding has relevant applications also in delegated (quantum) computation. We consider the protocol proposed in \cite{mantri_flow_2017}, whose security guarantees are based on Nayak's bound.
% We can strengthen their claim applying Lemma 2 and 3.
\bibliographystyle{IEEEtran}
\bibliography{lib}

@inproceedings{nayak_optimal_1999,
	title = {Optimal lower bounds for quantum automata and random access codes},
	doi = {10.1109/SFFCS.1999.814608},
	booktitle = {40th {Annual} {Symposium} on {Foundations} of {Computer} {Science} ({Cat}. {No}.{99CB37039})},
	author = {Nayak, A.},
	month = oct,
	year = {1999},
	pages = {369--376},
}

@article{hadiashar2023optimal,
  title={Optimal lower bounds for quantum learning via information theory},
  author={Hadiashar, Shima Bab and Nayak, Ashwin and Sinha, Pulkit},
  journal={IEEE Transactions on Information Theory},
  volume={70},
  number={3},
  pages={1876--1896},
  year={2023},
  publisher={IEEE}
}

@article{jain_new_2009,
	title = {New bounds on classical and quantum one-way communication complexity},
	volume = {410},
	issn = {0304-3975},
	url = {https://www.sciencedirect.com/science/article/pii/S0304397508007627},
	doi = {10.1016/j.tcs.2008.10.014},
	language = {en},
	number = {26},
	urldate = {2021-06-03},
	journal = {Theoretical Computer Science},
	author = {Jain, Rahul and Zhang, Shengyu},
	month = jun,
	year = {2009},
	pages = {2463--2477},
}

@article{arunachalam_optimal_2018,
	title = {Optimal {Quantum} {Sample} {Complexity} of {Learning} {Algorithms}},
	volume = {19},
	url = {http://jmlr.org/papers/v19/18-195.html},
	number = {71},
	journal = {Journal of Machine Learning Research},
	author = {Arunachalam, Srinivasan and Wolf, Ronald de},
	year = {2018},
	pages = {1--36},
}

@article{mantri_flow_2017,
	title = {Flow {Ambiguity}: {A} {Path} {Towards} {Classically} {Driven} {Blind} {Quantum} {Computation}},
	volume = {7},
	url = {https://link.aps.org/doi/10.1103/PhysRevX.7.031004},
	doi = {10.1103/PhysRevX.7.031004},
	number = {3},
	journal = {Phys. Rev. X},
	author = {Mantri, Atul and Demarie, Tommaso F. and Menicucci, Nicolas C. and Fitzsimons, Joseph F.},
	month = jul,
	year = {2017},
	pages = {031004},
}

@article{shannon_mathematical_1948,
	title = {A {Mathematical} {Theory} of {Communication}},
	volume = {27},
	copyright = {{\textcopyright} 1948 The Bell System Technical Journal},
	issn = {1538-7305},
	url = {https://onlinelibrary.wiley.com/doi/abs/10.1002/j.1538-7305.1948.tb01338.x},
	doi = {10.1002/j.1538-7305.1948.tb01338.x},
	language = {en},
	number = {3},
	urldate = {2021-06-18},
	journal = {Bell System Technical Journal},
	author = {Shannon, C. E.},
	year = {1948},
	pages = {379--423},
}

@article{schumacher_quantum_1995,
	title = {Quantum coding},
	volume = {51},
	url = {https://link.aps.org/doi/10.1103/PhysRevA.51.2738},
	doi = {10.1103/PhysRevA.51.2738},
	number = {4},
	journal = {Phys. Rev. A},
	author = {Schumacher, Benjamin},
	month = apr,
	year = {1995},
	pages = {2738--2747},
}

@article{jozsa_new_1994,
	title = {A {New} {Proof} of the {Quantum} {Noiseless} {Coding} {Theorem}},
	volume = {41},
	issn = {0950-0340},
	url = {https://doi.org/10.1080/09500349414552191},
	doi = {10.1080/09500349414552191},
	number = {12},
	urldate = {2021-06-18},
	journal = {Journal of Modern Optics},
	author = {Jozsa, Richard and Schumacher, Benjamin},
	month = dec,
	year = {1994},
	pages = {2343--2349},
}

@article{holevo_bounds_1973,
	title = {Bounds for the {Quantity} of {Information} {Transmitted} by a {Quantum} {Communication} {Channel}},
	volume = {9},
	number = {3},
	journal = {Problems of Information Transmission},
	author = {Holevo, Alexander},
	year = {1973},
	pages = {177--183},
}

@article{larose_overview_2019,
	title = {Overview and {Comparison} of {Gate} {Level} {Quantum} {Software} {Platforms}},
	volume = {3},
	url = {https://quantum-journal.org/papers/q-2019-03-25-130/},
	doi = {10.22331/q-2019-03-25-130},
	language = {en-GB},
	urldate = {2021-07-29},
	journal = {Quantum},
	author = {LaRose, Ryan},
	month = mar,
	year = {2019},
	pages = {130},
}

@article{harrow2009quantum,
  title={Quantum algorithm for linear systems of equations},
  author={Harrow, Aram W and Hassidim, Avinatan and Lloyd, Seth},
  journal={Physical review letters},
  volume={103},
  number={15},
  pages={150502},
  year={2009},
  publisher={APS}
}

@inproceedings{kerenidis2017quantum,
  title={Quantum Recommendation Systems},
  author={Kerenidis, Iordanis and Prakash, Anupam},
  booktitle={8th Innovations in Theoretical Computer Science Conference (ITCS 2017)},
  pages={49--1},
  year={2017},
  organization={Schloss Dagstuhl--Leibniz-Zentrum f{\"u}r Informatik}
}

@article{kerenidis2019q,
  title={q-means: A quantum algorithm for unsupervised machine learning},
  author={Kerenidis, Iordanis and Landman, Jonas and Luongo, Alessandro and Prakash, Anupam},
  journal={Advances in neural information processing systems},
  volume={32},
  year={2019}
}

@article{huang2021power,
  title={Power of data in quantum machine learning},
  author={Huang, Hsin-Yuan and Broughton, Michael and Mohseni, Masoud and Babbush, Ryan and Boixo, Sergio and Neven, Hartmut and McClean, Jarrod R},
  journal={Nature communications},
  volume={12},
  number={1},
  pages={2631},
  year={2021},
  publisher={Nature Publishing Group UK London}
}

@article{anshu2024survey,
  title={A survey on the complexity of learning quantum states},
  author={Anshu, Anurag and Arunachalam, Srinivasan},
  journal={Nature Reviews Physics},
  volume={6},
  number={1},
  pages={59--69},
  year={2024},
  publisher={Nature Publishing Group UK London}
}

@inproceedings{chen2022exponential,
  title={Exponential separations between learning with and without quantum memory},
  author={Chen, Sitan and Cotler, Jordan and Huang, Hsin-Yuan and Li, Jerry},
  booktitle={2021 IEEE 62nd Annual Symposium on Foundations of Computer Science (FOCS)},
  pages={574--585},
  year={2022},
  organization={IEEE}
}

@article{kostina2012fixed,
  title={Fixed-length lossy compression in the finite blocklength regime},
  author={Kostina, Victoria and Verd{\'u}, Sergio},
  journal={IEEE Transactions on Information Theory},
  volume={58},
  number={6},
  pages={3309--3338},
  year={2012},
  publisher={IEEE}
}

@inproceedings{broadbent_universal_2009,
	title = {Universal {Blind} {Quantum} {Computation}},
	doi = {10.1109/FOCS.2009.36},
	booktitle = {2009 50th {Annual} {IEEE} {Symposium} on {Foundations} of {Computer} {Science}},
	author = {Broadbent, Anne and Fitzsimons, Joseph and Kashefi, Elham},
	month = oct,
	year = {2009},
	pages = {517--526},
}

@inproceedings{cojocaru_qfactory_2019,
	address = {Cham},
	series = {Lecture {Notes} in {Computer} {Science}},
	title = {{QFactory}: {Classically}-{Instructed} {Remote} {Secret} {Qubits} {Preparation}},
	isbn = {978-3-030-34578-5},
	shorttitle = {{QFactory}},
	doi = {10.1007/978-3-030-34578-5_22},
	language = {en},
	booktitle = {Advances in {Cryptology} {\textendash} {ASIACRYPT} 2019},
	publisher = {Springer International Publishing},
	author = {Cojocaru, Alexandru and Colisson, L{\'e}o and Kashefi, Elham and Wallden, Petros},
	editor = {Galbraith, Steven D. and Moriai, Shiho},
	year = {2019},
	pages = {615--645},
}

@inproceedings{gheorghiu_computationally-secure_2019,
	title = {Computationally-{Secure} and {Composable} {Remote} {State} {Preparation}},
	doi = {10.1109/FOCS.2019.00066},
	booktitle = {2019 {IEEE} 60th {Annual} {Symposium} on {Foundations} of {Computer} {Science} ({FOCS})},
	author = {Gheorghiu, Alexandru and Vidick, Thomas},
	month = nov,
	year = {2019},
	pages = {1024--1033},
}

@article{aaronson_complexity-theoretic_2019,
	title = {Complexity-theoretic limitations on blind delegated quantum computation},
	url = {http://arxiv.org/abs/1704.08482},
	urldate = {2021-07-29},
	journal = {arXiv:1704.08482 [quant-ph]},
	author = {Aaronson, Scott and Cojocaru, Alexandru and Gheorghiu, Alexandru and Kashefi, Elham},
	month = feb,
	year = {2019},
}

@article{fitzsimons_unconditionally_2017,
	title = {Unconditionally verifiable blind quantum computation},
	volume = {96},
	url = {https://link.aps.org/doi/10.1103/PhysRevA.96.012303},
	doi = {10.1103/PhysRevA.96.012303},
	number = {1},
	journal = {Phys. Rev. A},
	author = {Fitzsimons, Joseph F. and Kashefi, Elham},
	month = jul,
	year = {2017},
	pages = {012303},
}

@article{childs_secure_2005,
	title = {Secure assisted quantum computation},
	volume = {5},
	issn = {1533-7146},
	url = {http://dx.doi.org/10.26421/QIC5.6},
	doi = {10.26421/qic5.6},
	number = {6},
	journal = {Quantum Information and Computation},
	author = {Childs, Andrew M.},
	month = sep,
	year = {2005},
}

@article{benedetti_parameterized_2019,
	title = {Parameterized quantum circuits as machine learning models},
	volume = {4},
	issn = {2058-9565},
	url = {https://doi.org/10.1088/2058-9565/ab4eb5},
	doi = {10.1088/2058-9565/ab4eb5},
	language = {en},
	number = {4},
	urldate = {2021-09-12},
	journal = {Quantum Science and Technology},
	author = {Benedetti, Marcello and Lloyd, Erika and Sack, Stefan and Fiorentini, Mattia},
	month = nov,
	year = {2019},
	pages = {043001},
}

@article{biamonte_quantum_2017,
	title = {Quantum machine learning},
	volume = {549},
	copyright = {2017 Macmillan Publishers Limited, part of Springer Nature. All rights reserved.},
	issn = {1476-4687},
	url = {https://www.nature.com/articles/nature23474},
	doi = {10.1038/nature23474},
	language = {en},
	number = {7671},
	urldate = {2021-09-12},
	journal = {Nature},
	author = {Biamonte, Jacob and Wittek, Peter and Pancotti, Nicola and Rebentrost, Patrick and Wiebe, Nathan and Lloyd, Seth},
	month = sep,
	year = {2017},
	pages = {195--202},
}

@article{ciliberto_quantum_2018,
	title = {Quantum machine learning: a classical perspective},
	volume = {474},
	shorttitle = {Quantum machine learning},
	url = {https://royalsocietypublishing.org/doi/10.1098/rspa.2017.0551},
	doi = {10.1098/rspa.2017.0551},
	number = {2209},
	urldate = {2021-09-12},
	journal = {Proceedings of the Royal Society A: Mathematical, Physical and Engineering Sciences},
	author = {Ciliberto, Carlo and Herbster, Mark and Ialongo, Alessandro Davide and Pontil, Massimiliano and Rocchetto, Andrea and Severini, Simone and Wossnig, Leonard},
	month = jan,
	year = {2018},
	pages = {20170551},
}

@article{bharti_noisy_2021,
	title = {Noisy intermediate-scale quantum ({NISQ}) algorithms},
	url = {http://arxiv.org/abs/2101.08448},
	urldate = {2021-09-12},
	journal = {arXiv:2101.08448 [cond-mat, physics:quant-ph]},
	author = {Bharti, Kishor and Cervera-Lierta, Alba and Kyaw, Thi Ha and Haug, Tobias and Alperin-Lea, Sumner and Anand, Abhinav and Degroote, Matthias and Heimonen, Hermanni and Kottmann, Jakob S. and Menke, Tim and Mok, Wai-Keong and Sim, Sukin and Kwek, Leong-Chuan and Aspuru-Guzik, Al{\'a}n},
	month = jan,
	year = {2021},
}

@article{adcock_advances_2015,
	title = {Advances in quantum machine learning},
	url = {http://arxiv.org/abs/1512.02900},
	urldate = {2021-09-12},
	journal = {arXiv:1512.02900 [quant-ph]},
	author = {Adcock, Jeremy and Allen, Euan and Day, Matthew and Frick, Stefan and Hinchliff, Janna and Johnson, Mack and Morley-Short, Sam and Pallister, Sam and Price, Alasdair and Stanisic, Stasja},
	month = dec,
	year = {2015},
}

@article{preskill_quantum_2018,
	title = {Quantum {Computing} in the {NISQ} era and beyond},
	volume = {2},
	url = {https://quantum-journal.org/papers/q-2018-08-06-79/},
	doi = {10.22331/q-2018-08-06-79},
	language = {en-GB},
	urldate = {2021-09-12},
	journal = {Quantum},
	author = {Preskill, John},
	month = aug,
	year = {2018},
	note = {Publisher: Verein zur F{\"o}rderung des Open Access Publizierens in den Quantenwissenschaften},
	pages = {79},
}

@article{valiant_theory_1984,
	title = {A theory of the learnable},
	volume = {27},
	issn = {0001-0782},
	url = {https://doi.org/10.1145/1968.1972},
	doi = {10.1145/1968.1972},
	number = {11},
	urldate = {2021-09-12},
	journal = {Communications of the ACM},
	author = {Valiant, L. G.},
	month = nov,
	year = {1984},
	pages = {1134--1142},
}

@article{feynman_simulating_1982,
	title = {Simulating physics with computers},
	volume = {21},
	issn = {1572-9575},
	url = {https://doi.org/10.1007/BF02650179},
	doi = {10.1007/BF02650179},
	language = {en},
	number = {6},
	urldate = {2021-09-12},
	journal = {International Journal of Theoretical Physics},
	author = {Feynman, Richard P.},
	month = jun,
	year = {1982},
	pages = {467--488},
}

@inproceedings{caro2024information,
  title={Information-theoretic generalization bounds for learning from quantum data},
  author={Caro, Matthias C and Gur, Tom and Rouz{\'e}, Cambyse and Franca, Daniel Stilck and Subramanian, Sathyawageeswar},
  booktitle={The Thirty Seventh Annual Conference on Learning Theory},
  pages={775--839},
  year={2024},
  organization={PMLR}
}

@inproceedings{bshouty_learning_1995,
	address = {New York, NY, USA},
	series = {{COLT} '95},
	title = {Learning {DNF} over the uniform distribution using a quantum example oracle},
	isbn = {978-0-89791-723-0},
	url = {https://doi.org/10.1145/225298.225312},
	doi = {10.1145/225298.225312},
	urldate = {2021-09-12},
	booktitle = {Proceedings of the eighth annual conference on {Computational} learning theory},
	publisher = {Association for Computing Machinery},
	author = {Bshouty, Nader H. and Jackson, Jeffrey C.},
	month = jul,
	year = {1995},
	pages = {118--127},
}

@inproceedings{wang_learning_2017,
	title = {Learning to {Model} the {Tail}},
	volume = {30},
	url = {https://proceedings.neurips.cc/paper/2017/file/147ebe637038ca50a1265abac8dea181-Paper.pdf},
	booktitle = {Advances in {Neural} {Information} {Processing} {Systems}},
	publisher = {Curran Associates, Inc.},
	author = {Wang, Yu-Xiong and Ramanan, Deva and Hebert, Martial},
	editor = {Guyon, I. and Luxburg, U. V. and Bengio, S. and Wallach, H. and Fergus, R. and Vishwanathan, S. and Garnett, R.},
	year = {2017},
}

@inproceedings{zhu_capturing_2014,
	title = {Capturing {Long}-{Tail} {Distributions} of {Object} {Subcategories}},
	doi = {10.1109/CVPR.2014.122},
	booktitle = {2014 {IEEE} {Conference} on {Computer} {Vision} and {Pattern} {Recognition}},
	author = {Zhu, Xiangxin and Anguelov, Dragomir and Ramanan, Deva},
	month = jun,
	year = {2014},
	note = {ISSN: 1063-6919},
	pages = {915--922},
}

@inproceedings{zhang_understanding_2017,
	title = {Understanding deep learning requires rethinking generalization},
	url = {https://openreview.net/forum?id=Sy8gdB9xx},
	booktitle = {5th {International} {Conference} on {Learning} {Representations}, {ICLR} 2017, {Toulon}, {France}, {April} 24-26, 2017, {Conference} {Track} {Proceedings}},
	publisher = {OpenReview.net},
	author = {Zhang, Chiyuan and Bengio, Samy and Hardt, Moritz and Recht, Benjamin and Vinyals, Oriol},
	year = {2017},
}

@article{arunachalam_guest_2017,
	title = {Guest {Column}: {A} {Survey} of {Quantum} {Learning} {Theory}},
	volume = {48},
	issn = {0163-5700},
	shorttitle = {Guest {Column}},
	url = {https://doi.org/10.1145/3106700.3106710},
	doi = {10.1145/3106700.3106710},
	number = {2},
	urldate = {2021-09-12},
	journal = {ACM SIGACT News},
	author = {Arunachalam, Srinivasan and de Wolf, Ronald},
	month = jun,
	year = {2017},
	pages = {41--67},
}

@article{preskill_quantum_2021,
	title = {Quantum computing 40 years later},
	url = {http://arxiv.org/abs/2106.10522},
	urldate = {2021-09-12},
	journal = {arXiv:2106.10522 [quant-ph]},
	author = {Preskill, John},
	month = jun,
	year = {2021},
	note = {arXiv: 2106.10522},
}

@article{benioff_computer_1980,
	title = {The computer as a physical system: {A} microscopic quantum mechanical {Hamiltonian} model of computers as represented by {Turing} machines},
	volume = {22},
	copyright = {1980 Plenum Publishing Corporation},
	issn = {1572-9613},
	shorttitle = {The computer as a physical system},
	url = {https://link.springer.com/article/10.1007/BF01011339},
	doi = {10.1007/BF01011339},
	language = {en},
	number = {5},
	urldate = {2021-09-12},
	journal = {Journal of Statistical Physics},
	author = {Benioff, Paul},
	month = may,
	year = {1980},
	note = {Company: Springer
Distributor: Springer
Institution: Springer
Label: Springer
Number: 5
Publisher: Kluwer Academic Publishers-Plenum Publishers},
	pages = {563--591},
}

@inproceedings{feldman,
author = {Feldman, Vitaly},
title = {Does Learning Require Memorization? A Short Tale about a Long Tail},
year = {2020},
isbn = {9781450369794},
publisher = {Association for Computing Machinery},
address = {New York, NY, USA},
url = {https://doi.org/10.1145/3357713.3384290},
doi = {10.1145/3357713.3384290},
booktitle = {Proceedings of the 52nd Annual ACM SIGACT Symposium on Theory of Computing},
pages = {954–959},
numpages = {6},
location = {Chicago, IL, USA},
series = {STOC 2020}
}

@article{wehner_quantum_2018,
author = {Stephanie Wehner  and David Elkouss  and Ronald Hanson },
title = {Quantum internet: A vision for the road ahead},
journal = {Science},
volume = {362},
number = {6412},
pages = {eaam9288},
year = {2018},
doi = {10.1126/science.aam9288},
URL = {https://www.science.org/doi/abs/10.1126/science.aam9288}
}

@article{datta_one-shot_2013,
	title = {One-{Shot} {Lossy} {Quantum} {Data} {Compression}},
	volume = {59},
	issn = {1557-9654},
	doi = {10.1109/TIT.2013.2283723},
	number = {12},
	journal = {IEEE Transactions on Information Theory},
	author = {Datta, Nilanjana and Renes, Joseph M. and Renner, Renato and Wilde, Mark M.},
	month = dec,
	year = {2013},
	note = {Conference Name: IEEE Transactions on Information Theory},
	pages = {8057--8076},
}

@article{datta_quantum--classical_2013,
	title = {Quantum-to-classical rate distortion coding},
	volume = {54},
	issn = {0022-2488},
	url = {https://aip.scitation.org/doi/10.1063/1.4798396},
	doi = {10.1063/1.4798396},
	number = {4},
	urldate = {2022-05-10},
	journal = {Journal of Mathematical Physics},
	author = {Datta, Nilanjana and Hsieh, Min-Hsiu and Wilde, Mark M. and Winter, Andreas},
	month = apr,
	year = {2013},
	note = {Publisher: American Institute of Physics},
	pages = {042201},
}

@article{wilde_quantum_2013,
	title = {Quantum {Rate}-{Distortion} {Coding} {With} {Auxiliary} {Resources}},
	volume = {59},
	issn = {1557-9654},
	doi = {10.1109/TIT.2013.2271772},
	number = {10},
	journal = {IEEE Transactions on Information Theory},
	author = {Wilde, Mark M. and Datta, Nilanjana and Hsieh, Min-Hsiu and Winter, Andreas},
	month = oct,
	year = {2013},
	note = {Conference Name: IEEE Transactions on Information Theory},
	pages = {6755--6773},
}

@article{caro2020quantum,
  title={Quantum learning Boolean linear functions wrt product distributions},
  author={Caro, Matthias C},
  journal={Quantum Information Processing},
  volume={19},
  number={6},
  pages={1--41},
  year={2020},
  publisher={Springer}
}
%\appendix

\appendix

%\section*{Appendix}
\renewcommand{\thesubsection}{\Alph{subsection}}

We introduce three useful lemmata which we employ in our results. %Combinatorial tools play a central role in our treatment. As such, the following two bounds on binomial coefficients are relevant:
\begin{lemma}[Binomial coefficient bounds]\label{lem:bin_coeff_bounds}
The following two bounds on the binomial coefficients hold true:
\begin{equation} \label{eqn:binomial_coefficient_bound}
    \binom{n}{k} \leq 2^{nH\left(\frac{k}{n}\right)}, \qquad \sum\limits_{i=0}^m \binom{n}{i} \leq 2^{nH\left(\frac{m}{n}\right)} ~ \forall m\leq \frac{n}{2},
\end{equation}
where, for $0\leq p \leq 1$, $H(p)$ is the entropy of a binary random variable that takes value $1$ with probability $p$ and value $0$ with probability $1-p$, that is $H(p):= -p\log(p) - (1-p)\log(1-p)$. 
\end{lemma}

Next, our second relevant lemma is an adaption of Theorem $15$ from~\cite{arunachalam_optimal_2018}. 

\begin{lemma}
\label{lem:qpsc}
For $\alpha \geq 0$, $ \epsilon \in [0,1/2)$, $\delta \in [0,1]$, assume that $I_{\text{Acc}}(\Sigma)\leq \alpha$ and $d_H(X,Z) \leq \epsilon n$ with probability at least $1-\delta$. Then,
\[
(1-\delta)[1-H(\epsilon)]n -{H(\delta)} \leq \alpha.
\]
\end{lemma}

\begin{proof}
[Proof (adapted from \cite{arunachalam_optimal_2018}).]
Since $I_{\text{Acc}}(\Sigma) \leq \alpha$, by definition of accessible information we have that:
\begin{equation*}
  I(X:Z) \leq \alpha.  
\end{equation*}
Now, let $V$ be the indicator random variable for the event that $d_H(X,Z) \leq \epsilon n$. Then $\Pr[V=1] \geq 1 - \delta$. Conditioning to $V=1$ and given $Z$, $X$ can range over a set of only $\sum_{i=0}^{\epsilon n}\binom{n}{i} \leq 2^{nH(\epsilon)}$ $n$-bit strings (using \eqref{eqn:binomial_coefficient_bound}), we find that:
\begin{equation*}
    H(X|Z,V=1) \leq n H(\epsilon)
\end{equation*}

We then lower bound $I(X:Z)$ as follows:

\begin{align*}
    I (X:Z) &= H (X) - H (X|Z) \\
&\geq H (X) - H (X|Z, V) - H (V) \\
&= H (X) - \Pr[V = 1] H(X|Z, V = 1) \\ - &\Pr[V = 0] H (X|Z, V = 0) - H (V) \\
&\geq n - (1 - \delta)H(\epsilon)n - \delta n - H(\delta)\\
&= (1 - \delta)(1 - H(\epsilon))n - H(\delta).
\end{align*}

Combining the two bounds we get the desired result:

\[ (1 - \delta)(1 - H(\epsilon))n - H(\delta) \leq \alpha. \]
\end{proof}
% \begin{lemma}[Entropy and counting]
% \label{lem:comb}
% For every integer $n$, we have that $\binom{n}{k} \leq 2^{nH(k/n)}$ and $\sum_{i=0}^m \binom{n}{i} \leq 2^{nH(m/n)}$ for all $m\leq n/2$.  
% \end{lemma}

Finally, we recall the last lemma, from~\cite{jain_new_2009}:

\begin{lemma}[Lemma 5 in~\cite{jain_new_2009}]
\label{lem:copies}
Let $m \geq 1$ be an integer and let $X,Y$ be correlated random variables. Let $\mu_x$ be the distribution of $Y|(X = x)$. Let $X', Y'$
represent joint random variables such that $X'$
is distributed identically
to $X$ and the distribution of $Y'
|(X' = x)$ is $\mu^m_x$
($m$ independent copies of $\mu_x$). Then,

\[
I(X': Y') \leq m  I(X : Y ).
\]

\end{lemma}

\end{document}